\documentclass[pra,twocolumn]{revtex4-1}%
\usepackage{amsfonts}
\usepackage{amsmath}
\usepackage{amssymb}
\usepackage{graphicx}
\usepackage[ colorlinks = true,              linkcolor = blue, urlcolor  =
blue,              citecolor = red,              anchorcolor = green,
]{hyperref}
\usepackage{xcolor}%
\providecommand{\U}[1]{\protect\rule{.1in}{.1in}}
\definecolor{mypink}{rgb}{0.858, 0.188, 0.478}
\newtheorem{theorem}{Theorem}

\newtheorem{corollary}{Corollary}

\newtheorem{definition}{Definition}

\newtheorem{lemma}{Lemma}

\newtheorem{proposition}{Proposition}
\newtheorem{remark}{Remark}

\newenvironment{proof}[1][Proof]{\noindent\textbf{#1.} }{\ \rule{0.5em}{0.5em}}
\allowdisplaybreaks
\begin{document}
\title{$\alpha$-Logarithmic negativity}
\author{Xin Wang}
\affiliation{Institute for Quantum Computing, Baidu Research, Beijing 100193, China}
\affiliation{Joint Center for Quantum Information and Computer Science, University of
Maryland, College Park, Maryland 20742, USA}
\email{wangxin73@baidu.com}

\author{Mark M.\ Wilde}
\email{mwilde@lsu.edu}
\affiliation{Hearne Institute for Theoretical Physics, Department of Physics and Astronomy,
and Center for Computation and Technology, Louisiana State University, Baton
Rouge, Louisiana 70803, USA}

\begin{abstract}
The logarithmic negativity of a bipartite quantum state is a widely employed
entanglement measure in quantum information theory, due to the fact that it is
easy to compute and serves as an upper bound on distillable entanglement. More
recently, the $\kappa$-entanglement of a bipartite state was shown to be the
first entanglement measure that is both easily computable and has a precise information-theoretic meaning, being equal to the exact entanglement cost of a bipartite quantum
state when the free operations are those that completely preserve the
positivity of the partial transpose [Wang and Wilde, Phys.~Rev.~Lett.~\textbf{125}(4):040502, July 2020]. In this paper, we provide a non-trivial
link between these two entanglement measures, by showing that they are the
extremes of an ordered family of $\alpha$-logarithmic negativity entanglement
measures, each of which is identified by a parameter $\alpha\in\left[
1,\infty\right]  $. In this family, the original logarithmic negativity is
recovered as the smallest with $\alpha=1$, and the $\kappa$-entanglement is
recovered as the largest with $\alpha=\infty$. We prove that the $\alpha
$-logarithmic negativity satisfies the following properties:\ 
entanglement monotone, normalization, faithfulness, and subadditivity. We also
prove that it is neither convex nor monogamous. Finally, we define the
$\alpha$-logarithmic negativity of a quantum channel as a generalization of
the notion for quantum states, and we show how to generalize many of the
concepts to arbitrary resource theories.

\end{abstract}
\date{\today}
\maketitle
\volumeyear{ }
\volumenumber{ }
\issuenumber{ }
\eid{ }
\startpage{1}
\endpage{102}

\section{Introduction}
Establishing a quantitative theory of entanglement has long been one of the
central concerns of quantum information theory \cite{H12,H06book,W17,Wat16}.
Starting with the original developments in \cite{BDSW96}, there now exists a
large body of knowledge on this topic \cite{H42007,Plenio2007}.

The traditional approaches to quantifying entanglement are the axiomatic
approach and the operational (information-theoretic or resource-theoretic)
approach. In the axiomatic approach, one identifies a list of desirable
properties that a measure of entanglement should possess. Most prominent among
these is that a measure of entanglement should not increase under the action
of a quantum channel realized by local operations and classical communication
(LOCC); if a measure of entanglement satisfies this property, it is called an
entanglement monotone. In the operational approach, one identifies a
meaningful information-processing task involving entanglement as a resource,
as well as some set of physical operations that are allowed for free. For
example, one could fix the free operations to be LOCC \cite{BDSW96}, separable
operations \cite{R97,VP98,PhysRevA.59.1070}, or operations that completely
preserve the positivity of the partial transpose (C-PPT-P operations)
\cite{R99}. Examples of information-processing tasks include entanglement
distillation \cite{BDSW96}, for which the goal is to use many copies of a
quantum state and free operations to produce as many high quality Bell states
as possible. In the opposite task, known as entanglement dilution
\cite{BDSW96,HHT01}, the goal is to use as few Bell states as possible, along
with LOCC, to produce as many high-fidelity copies of a quantum state as possible.

One of the most well known measures of entanglement is the logarithmic
negativity \cite{Vidal2002,Plenio2005b}, defined for a bipartite state
$\rho_{AB}$ as%
\begin{equation}
E_{N}(\rho_{AB})\equiv\log_{2}\left\Vert T_{B}(\rho_{AB})\right\Vert _{1},
\label{eq:log-neg}%
\end{equation}
where $T_{B}$ is the partial transpose map \cite{paulsen_2003}, defined for an
orthonormal basis $\{|i\rangle_{B}\}_{i}$ as%
\begin{equation}
T_{B}(Y_{AB})\equiv\sum_{i,j}\left(  I_{A}\otimes|i\rangle\langle
j|_{B}\right)  Y_{AB}\left(  I_{A}\otimes|i\rangle\langle j|_{B}\right)  ,
\label{eq:partial-transpose}%
\end{equation}
and $\left\Vert X\right\Vert _{1}\equiv\operatorname{Tr}[\sqrt{X^{\dag}X}]$
denotes the trace norm of an operator $X$. The logarithmic negativity
satisfies a number of properties, the most important of which is that it is an
entanglement monotone \cite{Vidal2002,Plenio2005b}. The widespread use of the
logarithmic negativity is due to the ease with which it can be computed, and
the fact that it provides an upper bound on the distillable entanglement of a
bipartite state \cite{HHH00,Vidal2002}. In this context, it should be
mentioned that an entanglement measure alternative to the logarithmic
negativity was recently proposed in \cite{WD16pra} (also known as the
max-Rains relative entropy \cite{WFD18}):\ it is easy to compute via
semi-definite programming, it satisfies all of the properties that the
logarithmic negativity does, and it provides a generally tighter upper bound
on distillable entanglement.

Another entanglement measure proposed in recent work is the $\kappa
$-entanglement of a quantum state \cite{WW18,WW20}, defined as follows:%
\begin{multline}
E_{\kappa}(\rho_{AB})\equiv\log_{2}\inf\{\operatorname{Tr}[S_{AB}]:\\
-T_{B}(S_{AB})\leq T_{B}(\rho_{AB})\leq T_{B}(S_{AB}),\ S_{AB}\geq0\}.
\end{multline}
The $\kappa$-entanglement can be computed easily by semi-definite programming
\cite{WW18,WW20}. In the context of the resource theory of entanglement, the
$\kappa$-entanglement measure can be regarded as a breakthrough:\ not only is
it easily computable and satisfies a number of desirable properties, but it
also has a precise information-theoretic interpretation as the exact entanglement cost of
a bipartite state $\rho_{AB}$ when C-PPT-P operations are allowed for free
(see \cite{WW18,WW20} for details and see \cite{APE03,MW08}\ for earlier work on this information-processing task).
No other entanglement measure is known to have these properties, which makes
$\kappa$-entanglement desirable from both the axiomatic and information-theoretic perspectives.

It is known that the $\kappa$-entanglement and logarithmic negativity coincide
for two-qubit states and for bosonic Gaussian states \cite{WW18,WW20}. This
reduction raises the question of whether there might be a deeper connection
between the two measures. \textit{In this paper, we show that this is indeed
the case, by defining a whole family of ordered entanglement measures that
interpolate between the logarithmic negativity, the smallest in the family,
and the }$\kappa$\textit{-entanglement, the largest in the family}. We call
each member of the family the $\alpha$-logarithmic negativity, where $\alpha$
is a parameter in the interval $\left[  1,\infty\right]  $. The $\alpha
$-logarithmic negativities are ordered, and the usual logarithmic negativity
is recovered when $\alpha=1$, whereas the $\kappa$-entanglement is recovered
when $\alpha=\infty$. In this sense, and following earlier traditions in
quantum information theory \cite{Renner2005,D09,Dat09,Tomamichel2015b}, we can
refer to the $\kappa$-entanglement alternatively as the
\textit{max-logarithmic negativity}.

Our results in this paper are related to the precedent of R\'enyi \cite{R61},
who established an interpolating family of entropic measures based on a
parameter$~\alpha$. However, it should be noted that we cannot interpret
$\alpha$-logarithmic negativity in terms of the traditional definitions of
quantum R\'enyi relative entropies \cite{P86,Mueller2013,WWY13}, for two reasons:

\begin{enumerate}
\item The normalization prefactor for $\alpha$-logarithmic negativity\ is
different from that employed for quantum R\'enyi relative entropies.

\item More importantly, the quantum R\'enyi relative entropies are functions
of quantum states, or more generally positive semi-definite operators, whereas
the $\alpha$-logarithmic negativity is based on an information measure that is
a function of a (not necessarily positive semi-definite) Hermitian operator.
\end{enumerate}

\noindent The second reason given above for the difference between $\alpha
$-logarithmic negativity and R\'enyi relative entropy is already apparent in
the definition of the logarithmic negativity in \eqref{eq:log-neg}, since the
argument is $T_{B}(\rho_{AB})$, which is a (not necessarily positive
semi-definite) Hermitian operator. Nevertheless, we discuss connections
between the $\alpha$-logarithmic negativity and sandwiched R\'enyi relative
entropy \cite{Mueller2013,WWY13} in the case that one allows for the first
argument of the sandwiched R\'enyi relative entropy to be a Hermitian operator.

We remark here that the technical development in our paper heavily relies on
results from \cite{B13monotone}\ and \cite{H16}. In particular, our proofs
that the $\alpha$-logarithmic negativities are ordered and are 
entanglement monotones strongly rely on the methods of \cite{B13monotone},
which therein were only applied to quantum states or positive semi-definite
operators. Here, we simply observe that the methods of \cite{B13monotone}%
\ apply when one of the operators is Hermitian, not unlike how the authors of
\cite{MHR2017}\ observed that the methods of \cite{B13monotone}\ apply more
generally to positive trace-preserving maps (rather than just the strict
subset of completely positive trace-preserving maps). Also, in order to
establish that the $\alpha$-logarithmic negativity can be computed by convex
optimization, we rely on a general theorem established in \cite{H16}. Thus,
given the benefits of the methods of \cite{B13monotone}\ and \cite{H16} for
establishing entanglement measures in quantum information, it seems fruitful
to continue the mathematical physics directions considered in
\cite{B13monotone} and \cite{H16} for future work, as has already been done in
several papers \cite{CFL16,CFL18,BST18,J18}.

We note here that other papers \cite{Calabrese2012,Lee2013,Calabrese_2014}
have pursued various generalizations of logarithmic negativity, but their
status as entanglement monotones remains unclear.

In the rest of the paper, we provide a detailed exposition of the $\alpha
$-logarithmic negativity. In particular, we define it in
Section~\ref{sec:def-alpha-log-neg}. In Section~\ref{sec:props-mu-nu}, we
detail several properties of the quantities underlying $\alpha$-logarithmic
negativity, and in Section~\ref{sec:ordering-alpha-log-neg}, we prove that the
$\alpha$-logarithmic negativities are ordered (i.e., monotone increasing with
respect to $\alpha\in\left[  1,\infty\right]  $). Then in
Section~\ref{sec:alpha-log-neg-limits}, we prove that the logarithmic
negativity and the $\kappa$-entanglement are special cases of the $\alpha
$-logarithmic negativity when $\alpha=1$ and $\alpha=\infty$, respectively. In
Section~\ref{sec:alpha-log-neg-props}, we establish several properties of the
$\alpha$-logarithmic negativity, including normalization, faithfulness,
subadditivity, and that it is a entanglement monotone. We also prove that
it is neither convex nor monogamous. In Section~\ref{sec:generalizations}, we
define the $\alpha$-logarithmic negativity of a quantum channel as a
generalization of the measure for states. Therein, we also discuss other
generalizations of the $\alpha$-logarithmic negativity. We finally conclude in
Section~\ref{sec:conclusion}\ with a brief summary and some open questions.

\section{$\alpha$-logarithmic negativity}

\label{sec:def-alpha-log-neg}In this section, we define the $\alpha
$-logarithmic negativity of a bipartite state $\rho_{AB}$.

First, let us define the following quantities, which are functions of
$\alpha\geq1$, a Hermitian operator $X\neq0$, and a positive semi-definite
operator $\sigma\neq0$:%
\begin{align}
\mu_{\alpha}(X\Vert\sigma)  &  \equiv\left\{
\begin{array}
[c]{cl}%
\left\Vert \sigma^{\frac{1-\alpha}{2\alpha}}X\sigma^{\frac{1-\alpha}{2\alpha}%
}\right\Vert _{\alpha} & \text{if}\ \operatorname{supp}(X)\subseteq
\operatorname{supp}(\sigma)\\
+\infty & \text{else}%
\end{array}
\right.  ,\label{eq:mu-def}\\
\nu_{\alpha}(X\Vert\sigma)  &  \equiv\log_{2}\mu_{\alpha}(X\Vert\sigma),
\label{eq:nu-def}%
\end{align}
where the $\alpha$-norm of an operator $Y$ is defined for $\alpha\geq1$ as%
\begin{align}
\left\Vert Y\right\Vert _{\alpha}  &  \equiv\left(  \operatorname{Tr}%
[\left\vert Y\right\vert ^{\alpha}]\right)  ^{1/\alpha},\\
\left\vert Y\right\vert  &  \equiv\sqrt{Y^{\dag}Y},
\end{align}
and the inverse $\sigma^{\left(  1-\alpha\right)  /2\alpha}$ is understood in
the generalized sense (i.e., taken on the support of $\sigma$). The definition
in \eqref{eq:mu-def} is consistent with the following limit:%
\begin{equation}
\mu_{\alpha}(X\Vert\sigma)=\lim_{\varepsilon\searrow0}\mu_{\alpha}%
(X\Vert\left(  1-\varepsilon\right)  \sigma+\varepsilon\theta),
\label{eq:limit-mu-alpha-support}%
\end{equation}
where $\theta$ is a positive definite state. A proof for the equality in
\eqref{eq:limit-mu-alpha-support} follows the same steps given in the proof of
\cite[Lemma~13]{Mueller2013}.

We define%
\begin{align}
\mu_{\infty}(X\Vert\sigma)  &  \equiv\lim_{\alpha\rightarrow\infty}\mu
_{\alpha}(X\Vert\sigma)\\
&  =\left\Vert \sigma^{-1/2}X\sigma^{-1/2}\right\Vert _{\infty}\\
&  =\inf\left\{  \lambda:-\lambda\sigma\leq X\leq\lambda\sigma\right\}  ,
\end{align}
and%
\begin{equation}
\nu_{\infty}(X\Vert\sigma)\equiv\log_{2}\mu_{\infty}(X\Vert\sigma).
\label{eq:nu-infty}%
\end{equation}
Both of the above formulas are defined as above in the case that
$\operatorname{supp}(X)\subseteq\operatorname{supp}(\sigma)$, and $\mu
_{\infty}(X\Vert\sigma)$\ and $\nu_{\infty}(X\Vert\sigma)$ are set to
$+\infty$ otherwise.

The function $\nu_{\alpha}(X\Vert\sigma)$ is related to the sandwiched R\'enyi
relative entropy $\widetilde{D}_{\alpha}(X\Vert\sigma)$%
\ \cite{Mueller2013,WWY13} of a Hermitian operator $X\neq0$ and a positive
semi-definite operator $\sigma\neq0$ as follows:%
\begin{equation}
\widetilde{D}_{\alpha}(X\Vert\sigma)=\frac{\alpha}{\alpha-1}\nu_{\alpha
}(X\Vert\sigma). \label{eq:sandwiched-Renyi}%
\end{equation}
Also, we have that%
\begin{align}
\nu_{\infty}(X\Vert\sigma)  &  =D_{\max}(X\Vert\sigma
)\label{eq:d-max-hermitian}\\
&  \equiv\log_{2}\left\Vert \sigma^{-1/2}X\sigma^{-1/2}\right\Vert _{\infty}\\
&  =\log\inf\left\{  \lambda:-\lambda\sigma\leq X\leq\lambda\sigma\right\}  ,
\end{align}
where $D_{\max}(X\Vert\sigma)$ is the max-relative entropy \cite{D09}\ of a
Hermitian operator $X\neq0$ and a positive semi-definite operator $\sigma
\neq0$. Note that $D_{\max}(X\Vert\sigma)=+\infty$ if $\operatorname{supp}%
(X)\not \subseteq \operatorname{supp}(\sigma)$. We also note here that both
the sandwiched R\'enyi relative entropy and max-relative entropy have only
been considered in prior work when $X$ is positive semi-definite, and so
\eqref{eq:sandwiched-Renyi} and \eqref{eq:d-max-hermitian} represent strict
generalizations of the previously considered definitions.

We are now ready to define the $\alpha$-logarithmic negativity of a bipartite
state as follows:

\begin{definition}
[$\alpha$-logarithmic negativity]\label{def:alpha-log-neg}Let $\rho_{AB}$ be a
bipartite state. We define its $\alpha$-logarithmic negativity as%
\begin{equation}
E_{N}^{\alpha}(\rho_{AB})\equiv\inf_{\sigma_{AB}\in\operatorname{PPT}(A:B)}%
\nu_{\alpha}(T_{B}(\rho_{AB})\Vert\sigma_{AB}),
\end{equation}
where $T_{B}$ is the partial transpose map, defined in
\eqref{eq:partial-transpose}, and the set $\operatorname{PPT}(A:B)$ is the set
of positive partial transpose states, defined as
\cite{P96,Horodecki19961,PH97,HHH98}%
\begin{multline}
\operatorname{PPT}(A:B)\equiv\\
\left\{  \sigma_{AB}:\sigma_{AB}\geq 0,\ T_{B}(\sigma_{AB})\geq0,\ \operatorname{Tr}%
[\sigma_{AB}]=1\right\}  .
\end{multline}

\end{definition}

It can be useful as a proof tool to take the infimum in
Definition~\ref{def:alpha-log-neg} over only the faithful states in
$\operatorname{PPT}(A:B)$, and the following lemma establishes that this is possible:

\begin{lemma}
\label{lem:inf-invertible-states}The following identity holds%
\begin{equation}
E_{N}^{\alpha}(\rho_{AB})=\inf_{\sigma_{AB}\in\operatorname{PPT}%
_{\operatorname{inv}}(A:B)}\nu_{\alpha}(T_{B}(\rho_{AB})\Vert\sigma_{AB}),
\end{equation}
where $\operatorname{PPT}_{\operatorname{inv}}(A:B)$ consists of the faithful
states in $\operatorname{PPT}(A:B)$:%
\begin{multline}
\operatorname{PPT}_{\operatorname{inv}}(A:B)\equiv\\
\left\{  \sigma_{AB}:\sigma_{AB} > 0,\ T_{B}(\sigma_{AB})>0,\ \operatorname{Tr}%
[\sigma_{AB}]=1\right\}  .
\end{multline}

\end{lemma}

\begin{proof}
This follows because the set $\operatorname{PPT}_{\operatorname{inv}}(A:B)$ is
dense in $\operatorname{PPT}(A:B)$ and the $\alpha$-norm is continuous, so
that the $\alpha$-norm of any state in $\operatorname{PPT}(A:B)$ can be
approximated arbitrarily well by that of a state in $\operatorname{PPT}%
_{\operatorname{inv}}(A:B)$.
\end{proof}

\section{Properties of the functions $\mu_{\alpha}(X\Vert\sigma)$ and
$\nu_{\alpha}(X\Vert\sigma)$}

\label{sec:props-mu-nu}In this section, we establish several properties of the
functions $\mu_{\alpha}(X\Vert\sigma)$ and $\nu_{\alpha}(X\Vert\sigma)$ that
are used in the rest of the paper. We also employ the following
abbreviations:\ CP stands for completely positive, TP stands for
trace-preserving, and CPTP\ stands for completely positive and trace preserving.

Throughout the section, $X\neq0$ is a Hermitian operator and $\sigma$ is a
positive definite operator. Following \cite{OZ99,B13monotone}, for $p\geq1$,
we define the following norm%
\begin{equation}
\left\Vert X\right\Vert _{p,\sigma}\equiv\left\Vert \sigma^{1/2p}%
X\sigma^{1/2p}\right\Vert _{p}.
\end{equation}
Note that%
\begin{equation}
\left\Vert X\right\Vert _{\infty,\sigma}\equiv\lim_{p\rightarrow\infty
}\left\Vert X\right\Vert _{p,\sigma}=\left\Vert X\right\Vert _{\infty}.
\end{equation}
We also define the following CP\ map \cite{B13monotone}:%
\begin{equation}
\Gamma_{\sigma}(X)\equiv\sigma^{1/2}X\sigma^{1/2},
\end{equation}
so that%
\begin{equation}
\Gamma_{\sigma}^{-1}(X)=\sigma^{-1/2}X\sigma^{-1/2}.
\end{equation}

\begin{lemma}
\label{lem:DP-mu-alpha}Let $X\neq0$ be a Hermitian operator, and let $\sigma$
be a positive definite operator. Let $\mathcal{P}$ be a positive and
trace-non-increasing map.\ Then the following inequality holds for all
$\alpha\geq1$:%
\begin{equation}
\nu_{\alpha}(X\Vert\sigma)\geq\nu_{\alpha}(\mathcal{P}(X)\Vert\mathcal{P}%
(\sigma)). \label{eq:nu-monotone-p-map}%
\end{equation}

\end{lemma}

\begin{proof}
This follows as a direct consequence of the reasoning given in
\cite[Theorem~6]{B13monotone} (see also \cite[Theorem~2]{MHR2017}). We repeat
the argument here for completeness.

Since $\alpha\geq1$ and the logarithm function is monotone increasing,
\eqref{eq:nu-monotone-p-map} is equivalent to%
\begin{equation}
\left\Vert \Gamma_{\sigma}^{-1}(X)\right\Vert _{\alpha,\sigma}\geq\left\Vert
\Gamma_{\mathcal{P}(\sigma)}^{-1}(\mathcal{P}(X))\right\Vert _{\alpha
,\mathcal{P}(\sigma)}.
\end{equation}
Observe that%
\begin{equation}
\Gamma_{\mathcal{P}(\sigma)}^{-1}(\mathcal{P}(X))=(\Gamma_{\mathcal{P}%
(\sigma)}^{-1}\circ\mathcal{P}\circ\Gamma_{\sigma})(\Gamma_{\sigma}^{-1}(X)).
\end{equation}
As a result,%
\begin{multline}
\left\Vert \Gamma_{\mathcal{P}(\sigma)}^{-1}(\mathcal{P}(X))\right\Vert
_{\alpha,\mathcal{P}(\sigma)}\\
\leq\left\Vert \Gamma_{\mathcal{P}(\sigma)}^{-1}\circ\mathcal{P}\circ
\Gamma_{\sigma}\right\Vert _{\left(  \alpha,\sigma\right)  \rightarrow\left(
\alpha,\mathcal{P}(\sigma)\right)  }\cdot\left\Vert \Gamma_{\sigma}%
^{-1}(X)\right\Vert _{\alpha,\sigma},
\end{multline}
where, for a map $\mathcal{M}$,%
\begin{equation}
\left\Vert \mathcal{M}\right\Vert _{\left(  \alpha,\sigma\right)
\rightarrow\left(  \alpha^{\prime},\sigma^{\prime}\right)  }\equiv\sup
_{Y\neq0}\frac{\left\Vert \mathcal{M}(Y)\right\Vert _{\alpha^{\prime}%
,\sigma^{\prime}}}{\left\Vert Y\right\Vert _{\alpha,\sigma}}.
\end{equation}
We also define%
\begin{equation}
\left\Vert \mathcal{M}\right\Vert _{\alpha\rightarrow\alpha^{\prime}}%
\equiv\left\Vert \mathcal{M}\right\Vert _{\left(  \alpha,I\right)
\rightarrow\left(  \alpha^{\prime},I\right)  }.
\end{equation}
So it suffices to establish that%
\begin{equation}
\left\Vert \Gamma_{\mathcal{P}(\sigma)}^{-1}\circ\mathcal{P}\circ
\Gamma_{\sigma}\right\Vert _{\left(  \alpha,\sigma\right)  \rightarrow\left(
\alpha,\mathcal{P}(\sigma)\right)  }\leq1. \label{eq:beigi-map-bnd}%
\end{equation}
By employing \cite[Theorem~4]{B13monotone}, it is only necessary to establish
\eqref{eq:beigi-map-bnd} for $\alpha=1$ and $\alpha=\infty$.

For $\alpha=1$, it follows that%
\begin{equation}
\left\Vert \Gamma_{\mathcal{P}(\sigma)}^{-1}\circ\mathcal{P}\circ
\Gamma_{\sigma}\right\Vert _{\left(  1,\sigma\right)  \rightarrow\left(
1,\mathcal{P}(\sigma)\right)  }=\left\Vert \mathcal{P}\right\Vert
_{1\rightarrow1}\leq1
\end{equation}
because $\mathcal{P}$ is positive and trace non-increasing. For $\alpha
=\infty$, it follows that%
\begin{multline}
\left\Vert \Gamma_{\mathcal{P}(\sigma)}^{-1}\circ\mathcal{P}\circ
\Gamma_{\sigma}\right\Vert _{(\infty,\sigma)\rightarrow(\infty,\mathcal{P}%
(\sigma))}\\
=\left\Vert \Gamma_{\mathcal{P}(\sigma)}^{-1}\circ\mathcal{P}\circ
\Gamma_{\sigma}\right\Vert _{\infty\rightarrow\infty}.
\end{multline}
Since $\Gamma_{\mathcal{P}(\sigma)}^{-1}\circ\mathcal{P}\circ\Gamma_{\sigma}$
is a positive map, by the Russo--Dye theorem (Corollary~2.3.8 of \cite{B07}),
we have that%
\begin{align}
&  \left\Vert \Gamma_{\mathcal{P}(\sigma)}^{-1}\circ\mathcal{P}\circ
\Gamma_{\sigma}\right\Vert _{\infty\rightarrow\infty}\nonumber\\
&  =\left\Vert (\Gamma_{\mathcal{P}(\sigma)}^{-1}\circ\mathcal{P}\circ
\Gamma_{\sigma})(I)\right\Vert _{\infty}\\
&  =\left\Vert (\Gamma_{\mathcal{P}(\sigma)}^{-1}\circ\mathcal{P}%
)(\sigma)\right\Vert _{\infty}\\
&  =\left\Vert \left(  \mathcal{P}(\sigma)\right)  ^{-1/2}\mathcal{P}%
(\sigma)\left(  \mathcal{P}(\sigma)\right)  ^{-1/2}\right\Vert _{\infty}\\
&  =1,
\end{align}
concluding the proof.
\end{proof}

\begin{lemma}
\label{lem:mu-alpha-cq}Let%
\begin{align}
Y_{XB}  & \equiv\sum_{x}p(x)|x\rangle\langle x|_{X}\otimes Y_{B}^{x},\\
\sigma_{XB}  & \equiv\sum_{x}q(x)|x\rangle\langle x|_{X}\otimes\sigma_{B}^{x},
\end{align}
where $\{Y_{B}^{x}\}_{x}$ is a set of Hermitian operators such that $Y_{B}%
^{x}\neq0$ for all $x$, $\left\{  p(x)\right\}  _{x}$ is a probability
distribution, $\{\sigma_{B}^{x}\}$ is a set of positive definite operators,
and $\{q(x)\}_{x}$ is a set of strictly positive reals. Then for $\alpha\geq
1$, we have that%
\begin{multline}
\nu_{\alpha}(Y_{XB}\Vert\sigma_{XB})\geq\sum_{x}p(x)\nu_{\alpha}(Y_{B}%
^{x}\Vert\sigma_{B}^{x})\\
+\left(  \frac{\alpha-1}{\alpha}\right)  D(p\Vert q),
\end{multline}
where $D(p\Vert q):=\sum_{x}p(x)\log_{2}(p(x)/q(x))$ is the classical relative entropy.
\end{lemma}

\begin{proof}
The case $\alpha=1$ follows directly because%
\begin{align}
\nu_{\alpha=1}(Y_{XB}\Vert\sigma_{XB})  & =\log_{2}\left\Vert Y_{XB}%
\right\Vert _{1}\\
& =\log_{2}\left[  \sum_{x}p(x)\left\Vert Y_{B}^{x}\right\Vert _{1}\right]
\\
& \geq\sum_{x}p(x)\log_{2}\left\Vert Y_{B}^{x}\right\Vert _{1}\\
& =\sum_{x}p(x)\nu_{\alpha=1}(Y_{B}^{x}\Vert\sigma_{B}^{x}),
\end{align}
where the inequality follows from concavity of the logarithm.

For the case $\alpha>1$, consider that%
\begin{align}
& \nu_{\alpha}(Y_{XB}\Vert\sigma_{XB})\nonumber\\
& =\log_{2}\mu_{\alpha}(Y_{XB}\Vert\sigma_{XB})\\
& =\frac{1}{\alpha}\log_{2}\left[  \mu_{\alpha}(Y_{XB}\Vert\sigma
_{XB})\right]  ^{\alpha}\\
& =\frac{1}{\alpha}\log_{2}\sum_{x}\left[  \mu_{\alpha}(p(x)Y_{B}^{x}\Vert
q(x)\sigma_{B}^{x})\right]  ^{\alpha}\\
& =\frac{1}{\alpha}\log_{2}\sum_{x}p(x)^{\alpha}q(x)^{1-\alpha}\left[
\mu_{\alpha}(Y_{B}^{x}\Vert\sigma_{B}^{x})\right]  ^{\alpha}\\
& =\frac{1}{\alpha}\log_{2}\sum_{x}p(x)\left(  \frac{p(x)}{q(x)}\right)
^{\alpha-1}\left[  \mu_{\alpha}(Y_{B}^{x}\Vert\sigma_{B}^{x})\right]
^{\alpha}\\
& \geq\frac{1}{\alpha}\sum_{x}p(x)\log_{2}\left[  \left(  \frac{p(x)}%
{q(x)}\right)  ^{\alpha-1}\left[  \mu_{\alpha}(Y_{B}^{x}\Vert\sigma_{B}%
^{x})\right]  ^{\alpha}\right]  \\
& =\left(  \frac{\alpha-1}{\alpha}\right)  \sum_{x}p(x)\log_{2}\left(
\frac{p(x)}{q(x)}\right)  \nonumber\\
& \qquad+\frac{1}{\alpha}\sum_{x}p(x)\log_{2}\left[  \mu_{\alpha}(Y_{B}%
^{x}\Vert\sigma_{B}^{x})\right]  ^{\alpha}\\
& =\left(\frac{\alpha-1}{\alpha}\right) D(p\Vert q)+\sum_{x}p(x)\nu_{\alpha}(Y_{B}^{x}%
\Vert\sigma_{B}^{x}).
\end{align}
The third equality follows from definitions and the fact that $\left\Vert
\sum_{y}|y\rangle\langle y|\otimes B^{y}\right\Vert _{\alpha}^{\alpha}%
=\sum_{y}\left\Vert B^{y}\right\Vert _{\alpha}^{\alpha}$ for a set
$\{B^{y}\}_{y}$ of operators. The sole inequality is a consequence of
concavity of the logarithm.
\end{proof}

\begin{lemma}
\label{lem:mu-alpha-to-1-ineq}Let $X\neq0$ be a Hermitian operator, and let
$\sigma$ be a positive definite operator. Then the following inequality holds
for all $\alpha\geq1$:%
\begin{equation}
\mu_{1}(X\Vert\sigma)=\left\Vert X\right\Vert _{1}\leq\mu_{\alpha}%
(X\Vert\sigma)\cdot(\operatorname{Tr}[\sigma])^{(\alpha-1)/\alpha}.
\end{equation}
Equivalently,%
\begin{equation}
\log_{2}\left\Vert X\right\Vert _{1}\leq\nu_{\alpha}(X\Vert\sigma)
+ \left(\frac{\alpha-1}{\alpha}\right) \log_2\operatorname{Tr}[\sigma].
\end{equation}

\end{lemma}

\begin{proof}
The proof of this lemma closely follows the proof of \cite[Theorem~5]%
{B13monotone}. Set $\alpha^{\prime}=\alpha/\left(  \alpha-1\right)  $.
Consider that%
\begin{align}
&  \mu_{\alpha}(X\Vert\sigma)\cdot(\operatorname{Tr}[\sigma])^{1/\alpha
^{\prime}}\nonumber\\
&  =\left\Vert \sigma^{\left(  1-\alpha\right)  /2\alpha}X\sigma^{\left(
1-\alpha\right)  /2\alpha}\right\Vert _{\alpha}\cdot(\operatorname{Tr}%
[\sigma])^{1/\alpha^{\prime}}\\
&  =\left\Vert \sigma^{-1/2\alpha^{\prime}}X\sigma^{-1/2\alpha^{\prime}%
}\right\Vert _{\alpha}\cdot(\operatorname{Tr}[\sigma])^{1/\alpha^{\prime}}\\
&  =\left\Vert \sigma^{1/2\alpha^{\prime}}\right\Vert _{2\alpha^{\prime}%
}\left\Vert \sigma^{-1/2\alpha^{\prime}}X\sigma^{-1/2\alpha^{\prime}%
}\right\Vert _{\alpha}\left\Vert \sigma^{1/2\alpha^{\prime}}\right\Vert
_{2\alpha^{\prime}}\\
&  \geq\left\Vert \sigma^{1/2\alpha^{\prime}}\sigma^{-1/2\alpha^{\prime}%
}X\sigma^{-1/2\alpha^{\prime}}\sigma^{1/2\alpha^{\prime}}\right\Vert _{1}\\
&  =\left\Vert X\right\Vert _{1},
\end{align}
where we used that $(\operatorname{Tr}[\sigma])^{1/\alpha^{\prime}}=\left\Vert
\sigma^{1/2\alpha^{\prime}}\right\Vert _{2\alpha^{\prime}}^{2}$ and the
inequality is a consequence of H\"{o}lder's inequality.
\end{proof}

\begin{lemma}
\label{lem:monotone-lemma}Let $X\neq0$ be a Hermitian operator, and let
$\sigma$ be a positive definite operator. Then the following inequality holds
for all $\beta>\alpha>1$:%
\begin{equation}
\left[  \frac{\mu_{\alpha}(X\Vert\sigma)}{\left\Vert X\right\Vert _{1}%
}\right]  ^{\frac{\alpha}{\alpha-1}}\leq\left[  \frac{\mu_{\beta}(X\Vert
\sigma)}{\left\Vert X\right\Vert _{1}}\right]  ^{\frac{\beta}{\beta-1}}.
\end{equation}
Equivalently,%
\begin{multline}
\frac{\alpha}{\alpha-1}\left[  \nu_{\alpha}(X\Vert\sigma)-\log_{2}\left\Vert
X\right\Vert _{1}\right]  \\
\leq\frac{\beta}{\beta-1}\left[  \nu_{\beta}(X\Vert\sigma)-\log_{2}\left\Vert
X\right\Vert _{1}\right]
\end{multline}

\end{lemma}

\begin{proof}
The proof of this lemma follows the proof of \cite[Theorem~7]{B13monotone}
closely. Since $1<\alpha<\beta$, it follows that $\frac{1}{\beta}<\frac
{1}{\alpha}<1$, so that there exists $\theta\in\left(  0,1\right)  $ such that%
\begin{equation}
\frac{1}{\alpha}=1-\theta+\frac{\theta}{\beta}.
\end{equation}
We then find that by simple manipulations that%
\begin{equation}
\theta\left(  \frac{\alpha}{\alpha-1}\right)  =\left(  \frac{\beta}{\beta
-1}\right)  .
\end{equation}
Consider from \cite[Corollary~3]{B13monotone} that the following inequality
holds%
\begin{equation}
\left\Vert \Gamma_{\sigma}^{-1}(X)\right\Vert _{\alpha,\sigma}\leq\left\Vert
\Gamma_{\sigma}^{-1}(X)\right\Vert _{1,\sigma}^{1-\theta}\left\Vert
\Gamma_{\sigma}^{-1}(X)\right\Vert _{\beta,\sigma}^{\theta}.
\end{equation}
Consider that%
\begin{align}
\left\Vert \Gamma_{\sigma}^{-1}(X)\right\Vert _{1,\sigma} &  =\left\Vert
X\right\Vert _{1},\\
\left\Vert \Gamma_{\sigma}^{-1}(X)\right\Vert _{\alpha,\sigma} &  =\mu
_{\alpha}(X\Vert\sigma),\\
\left\Vert \Gamma_{\sigma}^{-1}(X)\right\Vert _{\beta,\sigma} &  =\mu_{\beta
}(X\Vert\sigma),
\end{align}
from which we conclude that%
\begin{equation}
\mu_{\alpha}(X\Vert\sigma)\leq\left\Vert X\right\Vert _{1}^{1-\theta}\cdot
\mu_{\beta}(X\Vert\sigma)^{\theta},
\end{equation}
so that%
\begin{align}
\mu_{\alpha}(X\Vert\sigma)^{\frac{\alpha}{\alpha-1}} &  \leq\left\Vert
X\right\Vert _{1}^{\frac{\left(  1-\theta\right)  \alpha}{\alpha-1}}\cdot
\mu_{\beta}(X\Vert\sigma)^{\frac{\theta\alpha}{\alpha-1}}\\
&  =\left\Vert X\right\Vert _{1}^{\frac{\alpha}{\left(  \alpha-1\right)
}-\frac{\beta}{\left(  \beta-1\right)  }}\cdot\mu_{\beta}(X\Vert\sigma
)^{\frac{\beta}{\beta-1}}.
\end{align}
Rewriting this, we find that%
\begin{equation}
\left[  \frac{\mu_{\alpha}(X\Vert\sigma)}{\left\Vert X\right\Vert _{1}%
}\right]  ^{\frac{\alpha}{\alpha-1}}\leq\left[  \frac{\mu_{\beta}(X\Vert
\sigma)}{\left\Vert X\right\Vert _{1}}\right]  ^{\frac{\beta}{\beta-1}}.
\end{equation}
Taking a logarithm, we get that%
\begin{multline}
\frac{\alpha}{\alpha-1}\left[  \nu_{\alpha}(X\Vert\sigma)-\log_{2}\left\Vert
X\right\Vert _{1}\right]  \\
\leq\frac{\beta}{\beta-1}\left[  \nu_{\beta}(X\Vert\sigma)-\log_{2}\left\Vert
X\right\Vert _{1}\right]  ,
\end{multline}
concluding the proof.
\end{proof}

\begin{corollary}
\label{cor:ordering-mu-alpha}Let $X\neq0$ be a Hermitian operator, and let
$\sigma$ be a positive definite operator. Then the following inequality holds
for all $\beta>\alpha>1$:%
\begin{equation}
\mu_{\alpha}(X\Vert\sigma)\leq\mu_{\beta}(X\Vert\sigma).
\end{equation}
Equivalently,%
\begin{equation}
\nu_{\alpha}(X\Vert\sigma)\leq\nu_{\beta}(X\Vert\sigma).
\end{equation}

\end{corollary}

\begin{proof}
This follows easily from the fact that%
\begin{equation}
\frac{\beta}{\beta-1}\leq\frac{\alpha}{\alpha-1}%
\end{equation}
for $\beta>\alpha>1$ and by applying Lemma~\ref{lem:monotone-lemma}.
\end{proof}

\begin{lemma}
\label{lem:convex-mu-alpha}Let $X\neq0$ be a Hermitian operator, and let
$\sigma$ be a positive definite operator. Then for all $\alpha\geq1$, the
following function is convex:%
\begin{equation}
\sigma\mapsto\left[  \mu_{\alpha}(X\Vert\sigma)\right]  ^{\alpha}.
\end{equation}

\end{lemma}

\begin{proof}
Consider that%
\begin{align}
\left[  \mu_{\alpha}(X\Vert\sigma)\right]  ^{\alpha} &  =\left\Vert
\sigma^{\frac{1-\alpha}{2\alpha}}X\sigma^{\frac{1-\alpha}{2\alpha}}\right\Vert
_{\alpha}^{\alpha}\\
&  =\operatorname{Tr}\left[  \left(  \sigma^{\frac{1-\alpha}{2\alpha}}%
X\sigma^{\frac{1-\alpha}{\alpha}}X\sigma^{\frac{1-\alpha}{2\alpha}}\right)
^{\frac{\alpha}{2}}\right]  .
\end{align}
A general theorem \cite[Theorem~5.2]{H16} states that the following function
of positive definite operators $A$ and $B$ is jointly convex in $A$ and $B$:%
\begin{equation}
\operatorname{Tr}\left[  \left(  \left[  \mathcal{P}_{1}(A^{p})\right]
^{1/2}\mathcal{P}_{2}(B^{q})\left[  \mathcal{P}_{1}(A^{p})\right]
^{1/2}\right)  ^{s}\right]  ,
\end{equation}
for positive maps $\mathcal{P}_{1}$ and $\mathcal{P}_{2}$, $p,q\in\left[
-1,0\right]  $, and $s\geq0$. (For this statement, please consult Theorem~5.2
of \cite{H16} and the brief remarks stated before Theorem~5.3 therein). Then
we see that the convexity of $\left[  \mu_{\alpha}(X\Vert\sigma)\right]
^{\alpha}$ in $\sigma$ for $\alpha\geq1$ follows as a special case of
\cite[Theorem~5.2]{H16} by taking%
\begin{align}
\mathcal{P}_{1} &  =\operatorname{id},\\
\mathcal{P}_{2}(\cdot) &  =X\left(  \cdot\right)  X,\\
p &  =q=\left(  1-\alpha\right)  /\alpha,\\
s &  =\alpha/2,\\
A &  =B=\sigma,
\end{align}
concluding the proof.
\end{proof}

\section{Ordering of the $\alpha$-logarithmic negativity}

\label{sec:ordering-alpha-log-neg}Recall that the logarithmic negativity of a
bipartite state $\rho_{AB}$ is defined as \cite{Vidal2002,Plenio2005b}%
\begin{equation}
E_{N}(\rho_{AB})\equiv\log_{2}\left\Vert T_{B}(\rho_{AB})\right\Vert _{1}.
\end{equation}

\begin{proposition}
\label{prop:alpha-log-neg-to-log-neg}Let $\rho_{AB}$ be a bipartite quantum
state, and let $1\leq\alpha\leq\beta$. Then%
\begin{equation}
E_{N}(\rho_{AB})\leq E_{N}^{\alpha}(\rho_{AB})\leq E_{N}^{\beta}(\rho_{AB}).
\label{eq:alpha-log-neg->=log-neg}%
\end{equation}

\end{proposition}

\begin{proof}
Let $\sigma_{AB}$ be an arbitrary state in $\operatorname{PPT}%
_{\operatorname{inv}}(A:B)$. Applying Lemma~\ref{lem:mu-alpha-to-1-ineq}, we
find that%
\begin{align}
E_{N}(\rho_{AB})  &  =\log_{2}\left\Vert T_{B}(\rho_{AB})\right\Vert _{1}\\
&  \leq\nu_{\alpha}(T_{B}(\rho_{AB})\Vert\sigma_{AB}).
\end{align}
Since the inequality holds for all $\sigma_{AB}\in\operatorname{PPT}%
_{\operatorname{inv}}(A:B)$, by applying Lemma~\ref{lem:inf-invertible-states}
and an infimum, we conclude the first inequality in \eqref{eq:alpha-log-neg->=log-neg}.

To establish the second inequality in \eqref{eq:alpha-log-neg->=log-neg}, let
$\sigma_{AB}$ be an arbitrary state in $\operatorname{PPT}_{\operatorname{inv}%
}(A:B)$. Then applying Definition~\ref{def:alpha-log-neg}\ and
Corollary~\ref{cor:ordering-mu-alpha}, we find that%
\begin{align}
E_{N}^{\alpha}(\rho_{AB})  &  \leq\nu_{\alpha}(T_{B}(\rho_{AB})\Vert
\sigma_{AB})\\
&  \leq\nu_{\beta}(T_{B}(\rho_{AB})\Vert\sigma_{AB}).
\end{align}
Since the inequality holds for all $\sigma_{AB}\in\operatorname{PPT}%
_{\operatorname{inv}}(A:B)$, we conclude that $E_{N}^{\alpha}(\rho_{AB})\leq
E_{N}^{\beta}(\rho_{AB})$.
\end{proof}

\section{Limits of the $\alpha$-logarithmic negativity}

\label{sec:alpha-log-neg-limits}In this section, we consider two limits of the
$\alpha$-logarithmic negativity, when $\alpha\rightarrow1$ and $\alpha
\rightarrow\infty$. The former limit converges to the well known logarithmic
negativity \cite{Vidal2002,Plenio2005b} and the latter converges to the
$\kappa$-entanglement of a bipartite quantum state \cite{WW18,WW20}. Recall that
the $\kappa$-entanglement of a quantum state is defined as \cite{WW18,WW20}%
\begin{multline}
E_{\kappa}(\rho_{AB})\equiv\log_{2}\inf\{\operatorname{Tr}[S_{AB}%
]:\label{eq:kappa-ent}\\
-T_{B}(S_{AB})\leq T_{B}(\rho_{AB})\leq T_{B}(S_{AB}),\ S_{AB}\geq0\}.
\end{multline}

\begin{proposition}
\label{prop:alpha-1-log-neg}Let $\rho_{AB}$ be a bipartite quantum state. Then%
\begin{equation}
\lim_{\alpha\rightarrow1}E_{N}^{\alpha}(\rho_{AB})=E_{N}(\rho_{AB}).
\end{equation}

\end{proposition}

\begin{proof}
Let $\sigma_{AB}$ be a state in $\operatorname{PPT}_{\operatorname{inv}}%
(A:B)$. Then%
\begin{align}
&  \lim_{\alpha\rightarrow1}E_{N}^{\alpha}(\rho_{AB})\nonumber\\
&  =\lim_{\alpha\rightarrow1}\inf_{\sigma_{AB}\in\operatorname{PPT}%
_{\operatorname{inv}}(A:B)}\nu_{\alpha}(T_{B}(\rho_{AB})\Vert\sigma_{AB})\\
&  \leq\lim_{\alpha\rightarrow1}\nu_{\alpha}(T_{B}(\rho_{AB})\Vert\sigma
_{AB})\\
&  =\lim_{\alpha\rightarrow1}\log_{2}\left\Vert \sigma_{AB}^{\left(
1-\alpha\right)  /2\alpha}T_{B}(\rho_{AB})\sigma_{AB}^{\left(  1-\alpha
\right)  /2\alpha}\right\Vert _{\alpha}\\
&  =\log_{2}\left\Vert T_{B}(\rho_{AB})\right\Vert _{1}\\
&  =E_{N}(\rho_{AB}).
\end{align}
Combining with the inequality $E_{N}^{\alpha}(\rho_{AB})\geq E_{N}(\rho_{AB})$
from Proposition~\ref{prop:alpha-log-neg-to-log-neg}, we conclude the proof.
\end{proof}

\begin{definition}
[Max-logarithmic negativity]For $\rho_{AB}$ a bipartite quantum state, the
max-logarithmic negativity $E_{N}^{\max}(\rho_{AB})$ is defined as%
\begin{equation}
E_{N}^{\max}(\rho_{AB})\equiv\inf_{\sigma_{AB}\in\operatorname{PPT}(A:B)}%
\nu_{\infty}(T_{B}(\rho_{AB})\Vert\sigma_{AB}),
\end{equation}
and $\nu_{\infty}$ is defined in \eqref{eq:nu-infty}.
\end{definition}

\begin{proposition}
\label{prop:E-kappa-E-N-max}Let $\rho_{AB}$ be a bipartite quantum state. Then%
\begin{align}
E_{\kappa}(\rho_{AB}) =E_{N}^{\max}(\rho_{AB}) =\lim_{\alpha\rightarrow\infty
}E_{N}^{\alpha}(\rho_{AB}). \label{eq:e-kappa-e-max-N}%
\end{align}

\end{proposition}

\begin{proof}
We first prove the first equality in \eqref{eq:e-kappa-e-max-N}. Consider that
by the substitution $S_{AB}\rightarrow T_{B}(S_{AB})$ in \eqref{eq:kappa-ent},
we find that%
\begin{multline}
E_{\kappa}(\rho_{AB})=\log_{2}\inf\{\operatorname{Tr}[T_{B}(S_{AB})]:\\
-S_{AB}\leq T_{B}(\rho_{AB})\leq S_{AB},\ T_{B}(S_{AB})\geq0\}.
\end{multline}
Since $\operatorname{Tr}[T_{B}(S_{AB})]=\operatorname{Tr}[S_{AB}]$, it follows
that%
\begin{multline}
E_{\kappa}(\rho_{AB})=\log_{2}\inf\{\operatorname{Tr}[S_{AB}%
]:\label{eq:E-kappa-subst-proof}\\
-S_{AB}\leq T_{B}(\rho_{AB})\leq S_{AB},\ T_{B}(S_{AB})\geq0\}.
\end{multline}
From the condition $-S_{AB}\leq T_{B}(\rho_{AB})\leq S_{AB}$, it follows that
$-S_{AB}\leq S_{AB}$ and thus that $S_{AB}\geq0$. By approximation (the fact
that positive definite operators are dense in the set of positive
semi-definite ones), it follows that we can take the infimum over $S_{AB}>0$.
Now make the substitution $S_{AB}\rightarrow\mu\sigma_{AB}$ in
\eqref{eq:E-kappa-subst-proof}, where $\mu>0$ and $\sigma_{AB}\in
\operatorname{PPT}_{\operatorname{inv}}(A:B)$. Then we find that%
\begin{multline}
E_{\kappa}(\rho_{AB})=\log_{2}\inf\{\mu:\\
-\mu\sigma_{AB}\leq T_{B}(\rho_{AB})\leq\mu\sigma_{AB},\\
\sigma_{AB}\in\operatorname{PPT}_{\operatorname{inv}}(A:B)\}.
\end{multline}
The condition $-\mu\sigma_{AB}\leq T_{B}(\rho_{AB})\leq\mu\sigma_{AB}$ is
equivalent to $-\mu I_{AB}\leq\sigma_{AB}^{-1/2}T_{B}(\rho_{AB})\sigma
_{AB}^{-1/2}\leq\mu I_{AB}$, which is in turn equivalent to%
\begin{equation}
\left\Vert \sigma_{AB}^{-1/2}T_{B}(\rho_{AB})\sigma_{AB}^{-1/2}\right\Vert
_{\infty}\leq\mu.
\end{equation}
So then%
\begin{multline}
E_{\kappa}(\rho_{AB})=\log_{2}\inf\Big\{\mu:
\left\Vert \sigma_{AB}^{-1/2}T_{B}(\rho_{AB})\sigma_{AB}^{-1/2}\right\Vert
_{\infty}\leq\mu,\\ \sigma_{AB}\in\operatorname{PPT}_{\operatorname{inv}%
}(A:B)\Big\}.
\end{multline}
Then it follows that%
\begin{align}
&  E_{\kappa}(\rho_{AB})\nonumber\\
&  =\inf_{\sigma_{AB}\in\operatorname{PPT}_{\operatorname{inv}}(A:B)}\log
_{2}\left\Vert \sigma_{AB}^{-1/2}T_{B}(\rho_{AB})\sigma_{AB}^{-1/2}\right\Vert
_{\infty}\\
&  =\inf_{\sigma_{AB}\in\operatorname{PPT}_{\operatorname{inv}}(A:B)}%
\nu_{\infty}(T_{B}(\rho_{AB})\Vert\sigma_{AB}),
\end{align}
thus establishing \eqref{eq:e-kappa-e-max-N}.

Now we establish the second equality in \eqref{eq:e-kappa-e-max-N}. Consider
that%
\begin{align}
&  \lim_{\alpha\rightarrow\infty}E_{N}^{\alpha}(\rho_{AB})\nonumber\\
&  =\sup_{\alpha\in\lbrack1,\infty)}\inf_{\sigma_{AB}\in\operatorname{PPT}%
(A:B)}\nu_{\alpha}(T_{B}(\rho_{AB})\Vert\sigma_{AB})\\
&  =\inf_{\sigma_{AB}\in\operatorname{PPT}(A:B)}\sup_{\alpha\in\lbrack
1,\infty)}\nu_{\alpha}(T_{B}(\rho_{AB})\Vert\sigma_{AB})\\
&  =\inf_{\sigma_{AB}\in\operatorname{PPT}(A:B)}\nu_{\infty}(T_{B}(\rho
_{AB})\Vert\sigma_{AB})\\
&  =E_{N}^{\max}(\rho_{AB}).
\end{align}
The first equality follows from the ordering inequality from
Proposition~\ref{prop:alpha-log-neg-to-log-neg}, using which we can replace
$\lim_{\alpha\rightarrow\infty}$ with $\sup_{\alpha\in\lbrack1,\infty)}$. The
second (critical)\ equality above is a consequence of the ordering inequality
from Proposition~\ref{prop:alpha-log-neg-to-log-neg}, the quasi-convexity of
$\nu_{\alpha}(T_{B}(\rho_{AB})\Vert\sigma_{AB})$ with respect to $\sigma_{AB}$
(Proposition~\ref{prop:convex-opt}), and the convexity and compactness of the
set $\operatorname{PPT}(A:B)$. All of these properties allow for applying the
minimax theorem from \cite[Corollary~A.2]{MH11}, concluding the proof of the
second equality in \eqref{eq:e-kappa-e-max-N}.
\end{proof}

\bigskip Putting together previous results, we conclude the following:

\begin{proposition}
If $\rho_{AB}$ satisfies the condition $T_{B}(|T_{B}(\rho_{AB})|)\geq0$, then
all $\alpha$-logarithmic negativities are equal; i.e., the following equality
holds for all $\alpha\geq1$:%
\begin{equation}
E_{N}^{\alpha}(\rho_{AB})=E_{N}(\rho_{AB}). \label{eq:negativity-collapse}%
\end{equation}

\end{proposition}

\begin{proof}
It is known from \cite[Proposition~3]{WW18} that $E_{N}(\rho_{AB})=E_{\kappa
}(\rho_{AB})$ if $\rho_{AB}$ satisfies the condition
$T_{B}(|T_{B}(\rho_{AB})|)\geq0$. Then the equality in
\eqref{eq:negativity-collapse} follows as a consequence of the ordering
inequality from Proposition~\ref{prop:alpha-log-neg-to-log-neg}, as well as
Propositions~\ref{prop:alpha-1-log-neg}\ and \ref{prop:E-kappa-E-N-max}.
\end{proof}

\begin{remark}
\label{rem:binegativity}Since all pure states \cite{AdMVW02}, two-qubit states
\cite{Ishizaka2004a}, Werner states \cite{APE03}, and bosonic Gaussian states
\cite{APE03}\ satisfy the condition $T_{B}(|T_{B}(\rho_{AB})|)\geq0$, we
conclude that the equality in \eqref{eq:negativity-collapse} holds for such states.
\end{remark}

\section{Properties of the $\alpha$-logarithmic negativity}

\label{sec:alpha-log-neg-props}Here we prove that the $\alpha$-logarithmic
negativity obeys several fundamental properties, making it an interesting
entanglement measure to consider in quantum information theory. These
properties include the following:

\begin{enumerate}
\item Entanglement monotone under selective C-PPT-P channels (which includes selective LOCC channels as a special case).

\item It can be calculated by convex optimization.

\item Normalization on maximally entangled states.

\item Faithfulness.

\item Subadditivity.
\end{enumerate}

\noindent We also prove by counterexample that the $\alpha$-logarithmic
negativities are neither convex nor monogamous.

\subsection{Entanglement monotonicity under selective completely positive partial transpose preserving channels}

Let us first recall the class of completely positive partial transpose
preserving (C-PPT-P)\ channels \cite{R99,R01}, defined as bipartite channels
$\mathcal{N}_{AB\rightarrow A^{\prime}B^{\prime}}$ such that $\mathcal{N}%
_{AB\rightarrow A^{\prime}B^{\prime}}$ is CPTP\ and the map $T_{B^{\prime}%
}\circ\mathcal{N}_{AB\rightarrow A^{\prime}B^{\prime}}\circ T_{B}$ is CP.
Related to this, a C-PPT-P quantum instrument consists of the collection
$\{\mathcal{N}_{AB\rightarrow A^{\prime}B^{\prime}}^{x}\}_{x}$, where each
$\mathcal{N}_{AB\rightarrow A^{\prime}B^{\prime}}^{x}$ is CP, the map
$T_{B^{\prime}}\circ\mathcal{N}_{AB\rightarrow A^{\prime}B^{\prime}}^{x}\circ
T_{B}$ is CP, and the sum map $\sum_{x}\mathcal{N}_{AB\rightarrow A^{\prime
}B^{\prime}}^{x}$ is TP.\ It is well known that the set of C-PPT-P\ channels
contains the set of LOCC\ channels, as well as the set of separable channels
\cite{R99,R01}.

The following fundamental theorem establishes that the $\alpha$-logarithmic
negativities are entanglement monotones for all $\alpha\geq1$:

\begin{theorem}
[Entanglement monotone]\label{thm:ent-monotone}Let $\{\mathcal{N}%
_{AB\rightarrow A^{\prime}B^{\prime}}^{x}\}_{x}$ be a C-PPT-P\ quantum
instrument, and let $\rho_{AB}$ be a bipartite state. Then the $\alpha
$-logarithmic negativity is an entanglement monotone; i.e., the following
inequality holds for all $\alpha\geq1$:%
\begin{equation}
E_{N}^{\alpha}(\rho_{AB})\geq\sum_{x:p(x)>0}p(x)E_{N}^{\alpha}(\rho
_{A^{\prime}B^{\prime}}^{x}),\label{eq:ent-monotone-thm}%
\end{equation}
where%
\begin{align}
p(x) &  \equiv\operatorname{Tr}[\mathcal{N}_{AB\rightarrow A^{\prime}%
B^{\prime}}^{x}(\rho_{AB})],\\
\rho_{A^{\prime}B^{\prime}}^{x} &  \equiv\frac{1}{p(x)}\mathcal{N}%
_{AB\rightarrow A^{\prime}B^{\prime}}^{x}(\rho_{AB}).
\end{align}

\end{theorem}

\begin{proof}
Let $\sigma_{AB}$ be an arbitrary state in $\operatorname{PPT}%
_{\operatorname{inv}}(A:B)$. Let%
\begin{align}
\mathcal{N}_{AB\rightarrow A^{\prime}B^{\prime}X}^{T}(Y_{AB}) &  \equiv
\sum_{x}\mathcal{N}_{AB\rightarrow A^{\prime}B^{\prime}}^{T,x}(Y_{AB}%
)\otimes|x\rangle\langle x|_{X},\nonumber\\
\mathcal{N}_{AB\rightarrow A^{\prime}B^{\prime}}^{T,x}(Y_{AB}) &
\equiv(T_{B^{\prime}}\circ\mathcal{N}_{AB\rightarrow A^{\prime}B^{\prime}}%
^{x}\circ T_{B})(Y_{AB}).
\end{align}
Note that the map $\mathcal{N}_{AB\rightarrow A^{\prime}B^{\prime}X}^{T}$ is
completely positive and trace preserving, which is a consequence of each map
$\mathcal{N}_{AB\rightarrow A^{\prime}B^{\prime}}^{T,x}$ being CP\ and the sum
map $\sum_{x}\mathcal{N}_{AB\rightarrow A^{\prime}B^{\prime}}^{T,x}$ being TP.
Let
\begin{equation}
\sigma_{A^{\prime}B^{\prime}}^{x}\equiv\frac{1}{q(x)}\mathcal{N}%
_{AB\rightarrow A^{\prime}B^{\prime}}^{T,x}(\sigma_{AB}),
\end{equation}
where the probability distribution $\{q(x)\}_{x}$ is defined as%
\begin{equation}
q(x)\equiv\operatorname{Tr}[\mathcal{N}_{AB\rightarrow A^{\prime}B^{\prime}%
}^{T,x}(\sigma_{AB})].
\end{equation}
Note that $\sigma_{A^{\prime}B^{\prime}}^{x}\in\operatorname{PPT}(A^{\prime
}:B^{\prime})$ because%
\begin{equation}
\sigma_{A^{\prime}B^{\prime}}^{x}=\frac{1}{q(x)}\mathcal{N}_{AB\rightarrow
A^{\prime}B^{\prime}}^{T,x}(\sigma_{AB})\geq0,\label{eq:sigma-PPT-1}%
\end{equation}
since $\sigma_{AB}\geq0$ and $\mathcal{N}_{AB\rightarrow A^{\prime}B^{\prime}%
}^{T,x}$ is CP and%
\begin{equation}
T_{B^{\prime}}(\sigma_{A^{\prime}B^{\prime}}^{x})=\frac{1}{q(x)}%
\mathcal{N}_{AB\rightarrow A^{\prime}B^{\prime}}^{x}(T_{B}(\sigma_{AB}%
))\geq0,\label{eq:sigma-PPT-2}%
\end{equation}
since $T_{B}(\sigma_{AB})\geq0$ and $\mathcal{N}_{AB\rightarrow A^{\prime
}B^{\prime}}^{x}$ is CP. Also, $\operatorname{Tr}[\sigma_{A^{\prime}B^{\prime
}}^{x}]=1$ by definition. Then consider that%
\begin{align}
&  \nu_{\alpha}(T_{B}(\rho_{AB})\Vert\sigma_{AB})\nonumber\\
&  \geq\nu_{\alpha}(\mathcal{N}_{AB\rightarrow A^{\prime}B^{\prime}X}%
^{T}(T_{B}(\rho_{AB}))\Vert\mathcal{N}_{AB\rightarrow A^{\prime}B^{\prime}%
X}^{T}(\sigma_{AB}))\\
&  \geq\left(  \frac{\alpha-1}{\alpha}\right)  D(p\Vert q)+\sum_{x}%
p(x)\ \nu_{\alpha}(T_{B^{\prime}}(\rho_{A^{\prime}B^{\prime}}^{x})\Vert
\sigma_{A^{\prime}B^{\prime}}^{x})\\
&  \geq\sum_{x}p(x)\ \nu_{\alpha}(T_{B^{\prime}}(\rho_{A^{\prime}B^{\prime}%
}^{x})\Vert\sigma_{A^{\prime}B^{\prime}}^{x})\\
&  \geq\sum_{x}p(x)\ E_{N}^{\alpha}(\rho_{A^{\prime}B^{\prime}}^{x}).
\end{align}
The first inequality follows from Lemma~\ref{lem:DP-mu-alpha} (data
processing). The second inequality follows from the facts that%
\begin{align}
\mathcal{N}_{AB\rightarrow A^{\prime}B^{\prime}X}^{T}(T_{B}(\rho_{AB}))  &
=\sum_{x}p(x)|x\rangle\langle x|_{X}\otimes T_{B^{\prime}}(\rho_{A^{\prime
}B^{\prime}}^{x}),\nonumber\\
\mathcal{N}_{AB\rightarrow A^{\prime}B^{\prime}X}^{T}(\sigma_{AB})  &
=\sum_{x}q(x)|x\rangle\langle x|_{X}\otimes\sigma_{A^{\prime}B^{\prime}}^{x},
\end{align}
and by applying Lemma~\ref{lem:mu-alpha-cq}. The third inequality follows
because the classical relative entropy $D(p\Vert q)$ is non-negative.  The final inequality follows from
Definition~\ref{def:alpha-log-neg}.

Since the chain of inequalities holds for an arbitrary state $\sigma_{AB}
\in\operatorname{PPT}_{\operatorname{inv}}(A:B)$, we conclude \eqref{eq:ent-monotone-thm}.
\end{proof}

\bigskip

Note that Theorem~\ref{thm:ent-monotone} applies to the case in which the
C-PPT-P instrument consists of a single element, i.e., when the C-PPT-P
instrument is really just a {C-PPT-P} channel. We remark here that if the goal
is to establish monotonicity under the action of a C-PPT-P channel (and not
the more general case of an instrument), then the proof is slightly simpler
than above. For clarity and due to its brevity, we show the few steps
explicitly now. Let $\mathcal{N}_{AB\rightarrow A^{\prime}B^{\prime}}$ be a
C-PPT-P channel, and let $\sigma_{AB}\in\operatorname{PPT}_{\operatorname{inv}%
}(A\!:\!B)$. By definition, $\mathcal{N}_{AB\rightarrow A^{\prime}B^{\prime}%
}^{T}\equiv T_{B^{\prime}}\circ\mathcal{N}_{AB\rightarrow A^{\prime}B^{\prime
}}\circ T_{B}$ is a quantum channel, and $\mathcal{N}_{AB\rightarrow
A^{\prime}B^{\prime}}^{T}(\sigma_{AB})\in\operatorname{PPT}(A^{\prime
}\!:\!B^{\prime})$ by the same reasoning given in \eqref{eq:sigma-PPT-1} and
\eqref{eq:sigma-PPT-2}. Then consider that
\begin{align}
&  \nu_{\alpha}(T_{B}(\rho_{AB})\Vert\sigma_{AB})\nonumber\\
&  \geq\nu_{\alpha}(\mathcal{N}_{AB\rightarrow A^{\prime}B^{\prime}}^{T}%
(T_{B}(\rho_{AB}))\Vert\mathcal{N}_{AB\rightarrow A^{\prime}B^{\prime}}%
^{T}(\sigma_{AB}))\\
&  =\nu_{\alpha}(T_{B^{\prime}}(\mathcal{N}_{AB\rightarrow A^{\prime}%
B^{\prime}}(\rho_{AB}))\Vert\mathcal{N}_{AB\rightarrow A^{\prime}B^{\prime}%
}^{T}(\sigma_{AB}))\\
&  \geq E_{N}^{\alpha}(\mathcal{N}_{AB\rightarrow A^{\prime}B^{\prime}}%
(\rho_{AB})),
\end{align}
where the first inequality follows from Lemma~\ref{lem:DP-mu-alpha} and the
second from the fact that $\mathcal{N}_{AB\rightarrow A^{\prime}B^{\prime}%
}^{T}(\sigma_{AB})\in\operatorname{PPT}(A^{\prime}\!:\!B^{\prime})$ and
Definition~\ref{def:alpha-log-neg}. Since the inequalities hold for all
$\sigma_{AB}\in\operatorname{PPT}_{\operatorname{inv}}(A:B)$, we conclude
that
\begin{equation}
E_{N}^{\alpha}(\rho_{AB})\geq E_{N}^{\alpha}(\mathcal{N}_{AB\rightarrow
A^{\prime}B^{\prime}}(\rho_{AB})).
\end{equation}

\subsection{Convex optimization}

\begin{proposition}
\label{prop:convex-opt} Let $\rho_{AB}$ be a bipartite quantum state, and let
$\alpha\geq1$. Then the $\alpha$-logarithmic negativity $E_{N}^{\alpha}%
(\rho_{AB})$ can be calculated by convex optimization.
\end{proposition}

\begin{proof}
We can rewrite $E_{N}^{\alpha}(\rho_{AB})$ as follows:%
\begin{align}
&  E_{N}^{\alpha}(\rho_{AB})\nonumber\\
&  =\inf_{\sigma_{AB}\in\operatorname{PPT}(A:B)}\log_{2}\mu_{\alpha}%
(T_{B}(\rho_{AB})\Vert\sigma_{AB})\\
&  =\inf_{\sigma_{AB}\in\operatorname{PPT}(A:B)}\frac{1}{\alpha}\log
_{2}\left[  \mu_{\alpha}(T_{B}(\rho_{AB})\Vert\sigma_{AB})\right]  ^{\alpha}\\
&  =\frac{1}{\alpha}\log_{2}\inf_{\sigma_{AB}\in\operatorname{PPT}%
(A:B)}\left[  \mu_{\alpha}(T_{B}(\rho_{AB})\Vert\sigma_{AB})\right]  ^{\alpha
}.
\end{align}
The statement of the proposition then follows as a consequence of
Lemma~\ref{lem:convex-mu-alpha}.
\end{proof}

\subsection{Normalization}

\begin{proposition}
[Normalization]For a maximally entangled state $\Phi_{AB}^{d}$ of Schmidt rank
$d\geq2$:%
\begin{equation}
\Phi_{AB}^{d}\equiv\frac{1}{d}\sum_{i,j}|i\rangle\langle j|_{A}\otimes
|i\rangle\langle j|_{B},
\end{equation}
where $\{|i\rangle_{A}\}_{i}$ and $\{|i\rangle_{B}\}_{i}$ are orthonormal
bases, the following equality holds for all $\alpha\geq1$%
\begin{equation}
E_{N}^{\alpha}(\Phi_{AB}^{d})=\log_{2}d.
\end{equation}

\end{proposition}

\begin{proof}
This is a direct consequence of $E_{N}(\Phi_{AB}^{d})=E_{\kappa}(\Phi_{AB}%
^{d})=\log_{2}d$ \cite{WW18}, the ordering inequality in
Proposition~\ref{prop:alpha-log-neg-to-log-neg}, and
Proposition~\ref{prop:E-kappa-E-N-max}.
\end{proof}

\subsection{Faithfulness}

It is known that the logarithmic negativity is faithful, meaning that
$E_{N}(\rho_{AB})\geq0$ and $E_{N}(\rho_{AB})=0$ if and only if $\rho_{AB}%
\in\operatorname{PPT}(A:B)$. To see that $E_{N}(\rho_{AB})\geq0$, consider
that
\begin{align}
E_{N}(\rho_{AB})  &  = \log_{2} \left\Vert T_{B}(\rho_{AB}) \right\Vert _{1}\\
&  \geq\log_{2} \operatorname{Tr}[ T_{B}(\rho_{AB})]\\
&  = \log_{2} \operatorname{Tr}[ \rho_{AB}] = 0.
\end{align}
The implication $\rho_{AB}\in\operatorname{PPT}(A:B) \Rightarrow E_{N}%
(\rho_{AB})=0$ follows easily from the fact that $\left\Vert T_{B}(\rho_{AB})
\right\Vert _{1} = \operatorname{Tr}[ T_{B}(\rho_{AB})]$ for such states, and
the opposite implication was shown, e.g., in \cite[Proposition~5]{WW18}. Using
this, we can conclude faithfulness for the $\alpha$-logarithmic negativity:

\begin{proposition}
[Faithfulness]\label{prop:alpha-log-neg-faithful}Let $\rho_{AB}$ be a
bipartite quantum state, and let $\alpha\geq1$. Then $E_{N}^{\alpha}(\rho
_{AB})\geq0$ and $E_{N}^{\alpha}(\rho_{AB})=0$ if and only if $\rho_{AB}%
\in\operatorname{PPT}(A:B)$.
\end{proposition}

\begin{proof}
The inequality $E_{N}^{\alpha}(\rho_{AB})\geq0$ follows from
Proposition~\ref{prop:alpha-log-neg-to-log-neg} and the fact that $E_{N}%
(\rho_{AB})\geq0$.

To see the other statement, let $\rho_{AB}\in\operatorname{PPT}(A:B)$. Then we
can pick $\sigma_{AB}=T_{B}(\rho_{AB})\in\operatorname{PPT}(A:B)$, and we find
that%
\begin{align}
&  E_{N}^{\alpha}(\rho_{AB})\nonumber\\
&  \leq\nu_{\alpha}(T_{B}(\rho_{AB})\Vert\sigma_{AB})\\
&  =\nu_{\alpha}(T_{B}(\rho_{AB})\Vert T_{B}(\rho_{AB}))\\
&  =\log_{2}\left\Vert \left(  T_{B}(\rho_{AB})\right)  ^{\frac{1-\alpha
}{2\alpha}}T_{B}(\rho_{AB})\left(  T_{B}(\rho_{AB})\right)  ^{\frac{1-\alpha
}{2\alpha}}\right\Vert _{\alpha}\\
&  =\log_{2}\left\Vert \left(  T_{B}(\rho_{AB})\right)  ^{\frac{1}{\alpha}%
}\right\Vert _{\alpha}\\
&  =0.
\end{align}
We then conclude that $E_{N}^{\alpha}(\rho_{AB})=0$ if $\rho_{AB}%
\in\operatorname{PPT}(A:B)$.

Now suppose that $E_{N}^{\alpha}(\rho_{AB})=0$. Then this means that
$E_{N}(\rho_{AB})=0$ by Proposition~\ref{prop:alpha-log-neg-to-log-neg}, and
we conclude that $\rho_{AB}\in\operatorname{PPT}(A:B)$ as a consequence of the
faithfulness of logarithmic negativity.
\end{proof}

\subsection{Subadditivity}

In this section, we establish subadditivity of the $\alpha$-logarithmic
negativity $E_{N}^{\alpha}$:

\begin{proposition}
[Subadditivity]Let $\rho_{A_{1}B_{1}}$ and $\omega_{A_{2}B_{2}}$ be bipartite
states. Then the following subadditivity inequality holds for all $\alpha
\geq1$:%
\begin{equation}
E_{N}^{\alpha}(\rho_{A_{1}B_{1}}\otimes\omega_{A_{2}B_{2}})\leq E_{N}^{\alpha
}(\rho_{A_{1}B_{1}})+E_{N}^{\alpha}(\omega_{A_{2}B_{2}}),
\label{eq:subadd-alpha-log}%
\end{equation}
where the bipartition for $E_{N}^{\alpha}(\rho_{A_{1}B_{1}}\otimes
\omega_{A_{2}B_{2}})$ is understood to be $A_{1}A_{2}|B_{1}B_{2}$.
\end{proposition}

\begin{proof}
Let $\sigma_{A_{1}B_{1}}^{(1)}$ and $\sigma_{A_{2}B_{2}}^{(2)}$ be arbitrary
PPT\ states in $\operatorname{PPT}(A_{1}:B_{1})$ and $\operatorname{PPT}%
(A_{2}:B_{2})$, respectively. Then it follows that $\sigma_{A_{1}B_{1}}%
^{(1)}\otimes\sigma_{A_{2}B_{2}}^{(2)}\in\operatorname{PPT}(A_{1}A_{2}%
:B_{1}B_{2})$, so that%
\begin{align}
&  E_{N}^{\alpha}(\rho_{A_{1}B_{1}}\otimes\omega_{A_{2}B_{2}})\nonumber\\
&  \leq\nu_{\alpha}(T_{B_{1}}(\rho_{A_{1}B_{1}})\otimes T_{B_{2}}%
(\omega_{A_{2}B_{2}})\Vert\sigma_{A_{1}B_{1}}^{(1)}\otimes\sigma_{A_{2}B_{2}%
}^{(2)})\\
&  =\nu_{\alpha}(T_{B_{1}}(\rho_{A_{1}B_{1}})\Vert\sigma_{A_{1}B_{1}}%
^{(1)})+\nu_{\alpha}(T_{B_{2}}(\omega_{A_{2}B_{2}})\Vert\sigma_{A_{2}B_{2}%
}^{(2)}).
\end{align}
where we have exploited the additivity of $\nu_{\alpha}$ to establish the
equality. Since the inequality holds for arbitrary $\sigma_{A_{1}B_{1}}^{(1)}$
and $\sigma_{A_{2}B_{2}}^{(2)}$, the inequality in \eqref{eq:subadd-alpha-log}\ follows.
\end{proof}

It is not clear to us whether the opposite inequality (superadditivity)\ holds
in general. It is well known that the logarithmic negativity is additive
\cite{Vidal2002}, and it has been shown recently in \cite{WW18} that the
max-logarithmic negativity ($\kappa$-entanglement)\ is additive also. So by
Remark~\ref{rem:binegativity}, it follows that the $\alpha$-logarithmic
negativities are additive for the states mentioned there. However,
establishing additivity in general (or a counterexample) is a problem that we
leave open for future work.

\subsection{No convexity}

As a consequence of the counterexample given in \cite[Proposition~6]{WW18}\ in
addition to Remark~\ref{rem:binegativity}, it follows that the $\alpha
$-logarithmic negativity is not generally convex for any choice of $\alpha
\in\lbrack1,\infty]$. Indeed, by picking%
\begin{align}
\rho_{AB}^{1}  &  \equiv\Phi_{AB}^{2},\\
\rho_{AB}^{2}  &  \equiv\frac{1}{2}\left(  |00\rangle\langle00|_{AB}%
+|11\rangle\langle11|_{AB}\right)  ,\\
\overline{\rho}_{AB}  &  \equiv\frac{1}{2}\left(  \rho_{AB}^{1}+\rho_{AB}%
^{2}\right)  ,
\end{align}
we find for all $\alpha\in\left[  1,\infty\right]  $ that%
\begin{align}
E_{N}^{\alpha}(\rho_{AB}^{1})  &  =1,\\
E_{N}^{\alpha}(\rho_{AB}^{2})  &  =0,\\
E_{N}^{\alpha}(\overline{\rho}_{AB})  &  =\log_{2}\frac{3}{2},
\end{align}
which implies that%
\begin{equation}
E_{N}^{\alpha}(\overline{\rho}_{AB})>\frac{1}{2}\left(  E_{N}^{\alpha}%
(\rho_{AB}^{1})+E_{N}^{\alpha}(\rho_{AB}^{2})\right)  .
\end{equation}

\subsection{No monogamy}

An entanglement measure $E$ is said to be monogamous \cite{CKW00,T04,KWin04}%
\ if the following inequality holds for all tripartite states $\rho_{ABC}$:%
\begin{equation}
E(\rho_{A:B})+E(\rho_{A:C})\leq E(\rho_{A:BC}),
\end{equation}
where the bipartition is indicated by a colon.

As a consequence of the counterexample given in \cite[Proposition~7]{WW18}\ in
addition to Propositions~\ref{prop:alpha-log-neg-to-log-neg}
and~\ref{prop:E-kappa-E-N-max} and Remark~\ref{rem:binegativity}, it follows
that the $\alpha$-logarithmic negativity is not generally monogamous for any
choice of $\alpha\in\lbrack1,\infty]$. Indeed, consider the following state of
three qubits:%
\begin{equation}
|\psi\rangle_{ABC}\equiv\frac{1}{2}\left(  |000\rangle_{ABC}+|011\rangle
_{ABC}+\sqrt{2}|110\rangle_{ABC}\right)  .
\end{equation}
It was shown in \cite[Proposition~7]{WW18} that the following inequality holds%
\begin{equation}
E_{\kappa}(\psi_{A:B})+E_{\kappa}(\psi_{A:C})>E_{\kappa}(\psi_{A:BC}).
\label{eq:no-monogamy}%
\end{equation}
Since the reduced states $\psi_{AB}$ and $\psi_{AC}$ are two-qubit states, it
follows from Remark~\ref{rem:binegativity} and \eqref{eq:no-monogamy}\ for all
$\alpha\geq1$ that%
\begin{align}
E_{N}^{\alpha}(\psi_{A:B})+E_{N}^{\alpha}(\psi_{A:C})  &  >E_{\kappa}%
(\psi_{A:BC})\\
&  \geq E_{N}^{\alpha}(\psi_{A:BC}),
\end{align}
where the last inequality is a consequence of
Propositions~\ref{prop:alpha-log-neg-to-log-neg} and
\ref{prop:E-kappa-E-N-max}. So monogamy does not hold for any of the $\alpha
$-logarithmic negativities.

\section{Generalizations}

\label{sec:generalizations}

\subsection{$\alpha$-Logarithmic negativity of a quantum channel}

We can also generalize the notion of $\alpha$-logarithmic negativity from
bipartite quantum states to point-to-point quantum channels. Before doing so,
let us recall that the logarithmic negativity of a quantum channel
$\mathcal{N}_{A\rightarrow B}$ is defined as \cite{HW01}%
\begin{equation}
E_{N}(\mathcal{N})\equiv\log_{2}\left\Vert T_{B}\circ\mathcal{N}_{A\rightarrow
B}\right\Vert _{\diamond}.
\end{equation}
Recall that the diamond norm of a Hermitian preserving map $\mathcal{M}%
_{A\rightarrow B}$\ is defined as \cite{Kit97}%
\begin{equation}
\left\Vert \mathcal{M}_{A\rightarrow B}\right\Vert _{\diamond}\equiv\sup
_{\psi_{RA}}\left\Vert \mathcal{M}_{A\rightarrow B}(\psi_{RA})\right\Vert
_{1},
\end{equation}
with the optimization over pure bipartite states $\psi_{RA}$ such that the
reference system $R$ is isomorphic to the channel input system $A$. By
applying definitions, we see that we can write the logarithmic negativity of a
quantum channel as an optimized version of the logarithmic negativity of
quantum states:%
\begin{equation}
E_{N}(\mathcal{N})=\sup_{\psi_{RA}}E_{N}(\omega_{RB}),
\end{equation}
where $\omega_{RB}\equiv\mathcal{N}_{A\rightarrow B}(\psi_{RA})$. Note that
this kind of channel generalization of state measures is quite common in
quantum information theory \cite{BBCW13,TGW14b,TWW14,WFD18}.

Continuing in this spirit, we define the following:

\begin{definition}
[$\alpha$-log.~negativity of a channel]\label{def:alpha-log-neg-channel}The
$\alpha$-logarithmic negativity of a quantum channel is defined for
$\alpha\geq1$ as%
\begin{equation}
E_{N}^{\alpha}(\mathcal{N})=\sup_{\psi_{RA}}E_{N}^{\alpha}(\omega_{RB}),
\label{eq:alpha-log-neg-channel}%
\end{equation}
with $\omega_{RB}\equiv\mathcal{N}_{A\rightarrow B}(\psi_{RA})$.
\end{definition}

We could define it more generally with an optimization over mixed input states
$\psi_{RA}$ with unbounded reference system $R$. However, the maximal value is
always achieved by a pure bipartite input state with reference system $R$
isomorphic to the channel input system $A$, as a consequence of the
monotonicity inequality in Theorem~\ref{thm:ent-monotone}, the Schmidt
decomposition theorem, and the invariance of $E_{N}^{\alpha}$ with respect to
isometric channels acting on the reference system $R$ (this latter statement
itself being a consequence of Theorem~\ref{thm:ent-monotone}).

By the above observation and Remark~\ref{rem:binegativity}, it follows that%
\begin{equation}
E_{N}^{\alpha}(\mathcal{N})=E_{N}(\mathcal{N})
\label{eq:collapse-special-channels}%
\end{equation}
for all $\alpha\in\left[  1,\infty\right]  $ if $\mathcal{N}$ is qubit channel
(with qubit input and qubit output), if $\mathcal{N}$ is a unitary channel, or
if $\mathcal{N}$ is a Werner--Holevo channel \cite{WH02,FHMV04,LM15}. Recall
that a Werner--Holevo channel $\mathcal{W}_{A\rightarrow B}^{\left(
p,d\right)  }$ with parameters $p\in\left[  0,1\right]  $ and integer $d\geq2$
is defined as \cite{WH02,FHMV04,LM15}
\begin{equation}
\mathcal{W}_{A\rightarrow B}^{\left(  p,d\right)  }\equiv\left(  1-p\right)
\mathcal{W}_{A\rightarrow B}^{\left(  0,d\right)  }+p\mathcal{W}_{A\rightarrow
B}^{\left(  1,d\right)  },
\end{equation}
where the channels $\mathcal{W}_{A\rightarrow B}^{\left(  0,d\right)  }$ and
$\mathcal{W}_{A\rightarrow B}^{\left(  1,d\right)  }$ are defined as
\begin{align}
\mathcal{W}_{A\rightarrow B}^{\left(  0,d\right)  }(X_{A})  &  \equiv
\frac{\operatorname{Tr}[X_{A}]I_{B}+\operatorname{id}_{A\rightarrow B}%
(T_{A}(X_{A}))}{d+1},\\
\mathcal{W}_{A\rightarrow B}^{\left(  1,d\right)  }(X_{A})  &  \equiv
\frac{\operatorname{Tr}[X_{A}]I_{B}-\operatorname{id}_{A\rightarrow B}%
(T_{A}(X_{A}))}{d-1},
\end{align}
and $T_{A}$ denotes the partial transpose map on system $A$. Note that the
Choi state of the Werner--Holevo channel $\mathcal{W}_{A\rightarrow B}^{p,d}$
is a Werner state \cite{W89}:%
\begin{multline}
\mathcal{W}_{A\rightarrow B}^{\left(  p,d\right)  }(\Phi_{RA})\equiv\left(
1-p\right)  \frac{2}{d\left(  d+1\right)  }\Pi_{RB}^{\mathcal{S}}+\\
p\frac{2}{d\left(  d-1\right)  }\Pi_{RB}^{\mathcal{A}},
\end{multline}
where%
\begin{align}
\Pi_{RB}^{\mathcal{S}}  &  \equiv\left(  I_{RB}+F_{RB}\right)  /2,\\
\Pi_{RB}^{\mathcal{A}}  &  \equiv\left(  I_{RB}-F_{RB}\right)  /2,
\end{align}
and $F_{RB}$ is the unitary swap operator.

The claim stated after \eqref{eq:collapse-special-channels}\ follows easily
from Remark~\ref{rem:binegativity} for qubit-qubit channels, for unitary
channels because the output state is pure (and then applying
Remark~\ref{rem:binegativity}), and for Werner--Holevo channels by employing
their covariance symmetry \cite{LM15} and Theorem~\ref{thm:ent-monotone} to
conclude that the optimal state $\psi_{RA}$ in
\eqref{eq:alpha-log-neg-channel}\ is the maximally entangled state, so that
the output of the channel is a Werner state, from which we conclude the claim
by applying Remark~\ref{rem:binegativity}.\ For the Werner--Holevo channel, we
can in fact give the following equality for all $\alpha\geq1$
\cite{APE03,WW18}:%
\begin{equation}
E_{N}^{\alpha}\!\left(  \mathcal{W}_{A\rightarrow B}^{\left(  p,d\right)
}\right)  =\left\{
\begin{array}
[c]{cc}%
\log_{2}\left(  \frac{2}{d}\left[  2p-1\right]  +1\right)  & \text{if }%
p\geq\frac{1}{2}\\
0 & \text{if }p\leq\frac{1}{2}%
\end{array}
\right.  .
\end{equation}

Also, as a consequence of the ordering inequality in
Proposition~\ref{prop:alpha-log-neg-to-log-neg}\ and \cite[Theorem~27]{WW18},
the following equalities hold for all $\alpha\geq1$:%
\begin{align}
E_{N}^{\alpha}(\mathcal{L}_{\eta,N_{B}})  &  =\log_{2}\left(  \frac{1+\eta
}{\left(  1-\eta\right)  \left(  2N_{B}+1\right)  }\right)  ,\\
E_{N}^{\alpha}(\mathcal{A}_{G,N_{B}})  &  =\log_{2}\left(  \frac{G+1}{\left(
G-1\right)  \left(  2N_{B}+1\right)  }\right)  ,\\
E_{N}^{\alpha}(\mathcal{T}_{\xi})  &  =\log_{2}\left(  1/\xi\right)  ,
\end{align}
where $\mathcal{L}_{\eta,N_{B}}$ is a thermal channel with transmissivity
$\eta\in\left(  0,1\right)  $ and thermal photon number $N_{B}\in\left(
0,\eta/\left[  1-\eta\right]  \right)  $, $\mathcal{A}_{G,N_{B}}$ is an
amplifier channel with gain $G>1$ and thermal photon number $N_{B}\in\left(
0,1/\left[  G-1\right]  \right)  $, and $\mathcal{T}_{\xi}$ is an
additive-noise channel with noise variance $\xi\in\left(  0,1\right)  $. (See
\cite[Section~VII]{WW18} for more details of these channels.)

\subsection{Other positive but not completely positive maps}

We note briefly here that the main definition in this paper can be generalized
to other maps, besides the partial transpose map, that are positive but not
completely positive. Let $\mathcal{P}_{B}$ be a such a map. Then we can define
the set of free states as follows:%
\begin{equation}
\text{P}(A:B)\equiv\left\{  \sigma_{AB}:\sigma_{AB},\mathcal{P}_{B}%
(\sigma_{AB})\geq0,\operatorname{Tr}[\sigma_{AB}]=1\right\}  ,
\end{equation}
and we can define a generalized logarithmic negativity of a bipartite state
$\rho_{AB}$ as follows for $\alpha\geq1$:%
\begin{equation}
E_{P}^{\alpha}(\rho_{AB})\equiv\inf_{\sigma_{AB}\in\text{P}(A:B)}\nu_{\alpha
}(\mathcal{P}_{B}(\rho_{AB})\Vert\sigma_{AB}).
\end{equation}
The same notion can be generalized to channels as in
Definition~\ref{def:alpha-log-neg-channel}. Such a concept could be
interesting to explore further.

\subsection{Generalizations to other resource theories}

Recently, there has been a concerted effort to generalize concepts developed
in particular quantum resource theories to more general ones (see
\cite{CG18}\ for a recent review and
\cite{BG15,RKR15,Regula_2017,AHJ18,TRBLA19,TR19,LY19,LW19}\ for other papers
in this spirit). With this in mind, it seems fruitful to generalize the
concepts developed in this paper to arbitrary resource theories, beyond the
resource theory of entanglement.

To this end, consider a resource theory with a set $\mathcal{S}_{F}$\ of free
states and a set $\mathcal{O}_{F}$ of free operations (a quantum operation
being a completely positive, trace non-increasing map). Furthermore, suppose
that $\mathcal{P}$ is a Hermiticity-preserving and trace-preserving map
satisfying%
\begin{equation}
\mathcal{F}\circ\mathcal{P}=\mathcal{P}\circ\mathcal{F},
\label{eq:gen-res-theory-commuting-prop}%
\end{equation}
for all free operations $\mathcal{F\in O}_{F}$, as well as%
\begin{equation}
\mathcal{P}(\sigma)\in\mathcal{S}_{F}\quad\text{if\quad}\sigma\in
\mathcal{S}_{F}. \label{eq:faithfulness-gen-res-1st-way}%
\end{equation}
(It should be clear that the map $\mathcal{P}$ mentioned above should
generalize the partial transpose operation considered earlier.) We could
demand further that%
\begin{equation}
\sigma\in\mathcal{S}_{F}\quad\text{if\quad}\mathcal{P}(\sigma)\in
\mathcal{S}_{F}. \label{eq:for-faithfulness-gen-res}%
\end{equation}

Then for $\alpha\geq1$, we define the $\alpha$-logarithmic resourcefulness of
a state $\rho$ with respect to $\mathcal{P}$\ as%
\begin{equation}
R_{\alpha}^{\mathcal{P}}(\rho)\equiv\inf_{\sigma\in\mathcal{S}_{F}}\nu
_{\alpha}(\mathcal{P}(\rho)\Vert\sigma).
\label{eq:def-alpha-log-resourcefulness}%
\end{equation}

This measure satisfies many of the properties that the $\alpha$-logarithmic
negativity does. By the same proof given for
Proposition~\ref{prop:alpha-log-neg-to-log-neg}, we have that the ordering
property is satisfied%
\begin{equation}
R_{\alpha}^{\mathcal{P}}(\rho)\leq R_{\beta}^{\mathcal{P}}(\rho),
\label{eq:ordering-gen-res}%
\end{equation}
for all states $\rho$ and for $1\leq\alpha\leq\beta$.

We also have the following regarding faithfulness:

\begin{proposition}
Fix $\alpha\geq1$. The $\alpha$-logarithmic resourcefulness is faithful,
meaning that $R_{\alpha}^{\mathcal{P}}(\rho)\geq0$ and $R_{\alpha
}^{\mathcal{P}}(\rho)=0$ if $\rho\in\mathcal{S}_{F}$. If
\eqref{eq:for-faithfulness-gen-res} holds and $\alpha>1$, then $\rho
\in\mathcal{S}_{F}$ if $R_{\alpha}^{\mathcal{P}}(\rho)=0$.
\end{proposition}

\begin{proof}
The proof follows along the same lines as that given for
Proposition~\ref{prop:alpha-log-neg-faithful}. To see that $R_{\alpha
}^{\mathcal{P}}(\rho)\geq0$ for an arbitrary state $\rho$, consider that%
\begin{align}
R_{\alpha}^{\mathcal{P}}(\rho)  &  \geq R_{1}^{\mathcal{P}}(\rho)\\
&  =\log_{2}\left\Vert \mathcal{P}(\rho)\right\Vert _{1}\\
&  \geq\log_{2}\operatorname{Tr}[\mathcal{P}(\rho)]\\
&  =\log_{2}\operatorname{Tr}[\rho]\\
&  =0.
\end{align}
For the first inequality, we used the ordering inequality in
\eqref{eq:ordering-gen-res}. The second inequality follows because $\left\Vert
X\right\Vert _{1}\geq\operatorname{Tr}[X]$ for an arbitrary square operator
$X$. The second equality exploits the assumption that $\mathcal{P}$ is trace preserving.

If $\rho\in\mathcal{S}_{F}$, then pick $\sigma=\mathcal{P}(\rho)\in
\mathcal{S}_{F}$ (following from \eqref{eq:faithfulness-gen-res-1st-way}), and
we find that%
\begin{align}
R_{\alpha}^{\mathcal{P}}(\rho)  &  \leq\nu_{\alpha}(\mathcal{P}(\rho
)\Vert\sigma)\\
&  =\nu_{\alpha}(\mathcal{P}(\rho)\Vert\mathcal{P}(\rho))\\
&  =\log_{2}\left\Vert \left[  \mathcal{P}(\rho)\right]  ^{1/\alpha
}\right\Vert _{\alpha}\\
&  =0.
\end{align}

If \eqref{eq:for-faithfulness-gen-res} holds and $R_{\alpha}^{\mathcal{P}%
}(\rho)=0$, then the ordering inequality in \eqref{eq:ordering-gen-res}
implies that $R_{1}^{\mathcal{P}}(\rho)=\log_{2}\left\Vert \mathcal{P}%
(\rho)\right\Vert _{1}=0$, which means that $\left\Vert \mathcal{P}%
(\rho)\right\Vert _{1}=1$. Let a Jordan--Hahn decomposition of $\mathcal{P}%
(\rho)$ be $\mathcal{P}(\rho)=P-N$ (i.e., $P,N\geq0$ and $PQ=0$). Consider
that $\operatorname{Tr}[\mathcal{P}(\rho)]=\operatorname{Tr}%
[P]-\operatorname{Tr}[N]=1$. Also, $1=\left\Vert \mathcal{P}(\rho)\right\Vert
_{1}=\operatorname{Tr}[P]+\operatorname{Tr}[N]$. Subtracting these equations
gives that $\operatorname{Tr}[N]=0\Rightarrow N=0$. Then $\mathcal{P}%
(\rho)\geq0$ and $\operatorname{Tr}[\mathcal{P}(\rho)]=1$. Since we know by
assumption that $R_{\alpha}^{\mathcal{P}}(\rho)=0$ for some $\alpha>1$ and
$\mathcal{P}(\rho)$ is a quantum state, from the relation in
\eqref{eq:sandwiched-Renyi}, the definition in
\eqref{eq:def-alpha-log-resourcefulness},\ and the equality conditions for
sandwiched R\'enyi relative entropy \cite{WWY13,Mueller2013,B13monotone}, it
follows that $\mathcal{P}(\rho)\in\mathcal{S}_{F}$, and so
\eqref{eq:for-faithfulness-gen-res} implies that $\rho\in\mathcal{S}_{F}$.
\end{proof}

\bigskip

Let $\{\mathcal{F}^{x}\}_{x}$ denote a free quantum instrument, i.e., such
that $\mathcal{F}^{x}\in\mathcal{O}_{F}$ and $\sum_{x}\mathcal{F}^{x}$ is a
quantum channel. The most prominent property of $\alpha$-logarithmic
resourcefulness is that it is monotone under the action of a free quantum
instrument, in the following sense:%
\begin{equation}
R_{\alpha}^{\mathcal{P}}(\rho)\geq\sum_{x:p(x)>0}p(x)\ R_{\alpha}%
^{\mathcal{P}}(\rho^{x}),
\end{equation}
where%
\begin{align}
p(x) &  \equiv\operatorname{Tr}[\mathcal{F}^{x}(\rho)],\\
\rho^{x} &  \equiv\frac{1}{p(x)}\mathcal{F}^{x}(\rho).
\end{align}
The proof of this inequality is nearly identical to that given for
Theorem~\ref{thm:ent-monotone}, with the main change being that we use the
defining property in \eqref{eq:gen-res-theory-commuting-prop} in the proof of
Theorem~\ref{thm:ent-monotone}.\ The proof is simpler in the case that one is
interested in establishing monotonicity under a free quantum channel rather
than a free quantum instrument, and so we detail it briefly here for clarity.
Let $\sigma\in\mathcal{S}_{F}$ and $\mathcal{F}\in\mathcal{C}_{F}$, where
$\mathcal{C}_{F}$ denotes the set of free channels. Then%
\begin{align}
\nu_{\alpha}(\mathcal{P}(\rho)\Vert\sigma) &  \geq\nu_{\alpha}((\mathcal{F}%
\circ\mathcal{P})(\rho)\Vert\mathcal{F}(\sigma))\\
&  =\nu_{\alpha}((\mathcal{P}\circ\mathcal{F})(\rho)\Vert\mathcal{F}%
(\sigma))\\
&  \geq R_{\alpha}^{\mathcal{P}}(\mathcal{F}(\rho)),
\end{align}
where the first inequality employs Lemma~\ref{lem:DP-mu-alpha} and the sole
equality employs \eqref{eq:gen-res-theory-commuting-prop}. Since the
inequality holds for all $\sigma\in\mathcal{S}_{F}$, we conclude that%
\begin{equation}
R_{\alpha}^{\mathcal{P}}(\rho)\geq R_{\alpha}^{\mathcal{P}}(\mathcal{F}%
(\rho)).
\end{equation}

We can also define the $\alpha$-logarithmic resourcefulness of a quantum
channel $\mathcal{N}$ for all $\alpha\geq1$ and with respect to $\mathcal{P}%
$\ as%
\begin{equation}
R_{\alpha}^{\mathcal{P}}(\mathcal{N})\equiv\sup_{\rho}R_{\alpha}^{\mathcal{P}%
}(\mathcal{N}(\rho)).
\end{equation}
We think it is interesting to explore particular instantiations of this
resourcefulness measure for particular resource theories, but we leave this
for future work.

\section{Conclusion}

\label{sec:conclusion}In summary, we have defined an ordered family of
$\alpha$-logarithmic negativity entanglement measures that interpolate between
the logarithmic negativity \cite{Vidal2002,Plenio2005b}\ and the $\kappa
$-entanglement \cite{WW18,WW20}.\ We proved that this family of entanglement
measures satisfies LOCC monotonicity, normalization, faithfulness,
subadditivity, and can be computed by convex optimization. The proofs of these
properties built upon prior results from \cite{B13monotone} and \cite{H16}. We
also proved that it is generally not convex nor is it monogamous. Finally, we
defined the $\alpha$-logarithmic negativity of a quantum channel as a
generalization of the measure for bipartite states, and we showed how to
generalize many of the concepts to arbitrary resource theories.

Going forward from here, we think that it is worthwhile to establish the
definition and properties of the $\alpha$-logarithmic negativity of quantum
states in the von Neumann algebra setting. The tools developed in
\cite{J18}\ should be useful for this task. Note here that we do think that it
is necessary to use the approach of \cite{J18} over that given in
\cite{BST18}, because our definition of $\alpha$-logarithmic negativity
requires working with a Hermitian operator and a positive semi-definite one.
Since \cite{J18}\ builds upon the approach given in \cite{B13monotone} for the
finite-dimensional case, we suspect that the methods of \cite{J18} should lead
to a cogent notion of $\alpha$-logarithmic negativity of quantum states in the
von Neumann algebra setting. One could also define and explore the $\alpha
$-logarithmic negativity of a quantum channel  in the von Neumann algebra
setting, generalizing Definition~\ref{def:alpha-log-neg-channel}.

As the logarithmic negativity is frequently used to study entanglement in
quantum many-body physics \cite{Calabrese2012,Lee2013,Castelnovo2013} and
holographic field theories \cite{Rangamani2014,Chaturvedi2018}, the $\alpha
$-logarithmic negativity may also shed light on these areas.

It is an interesting open question to determine whether $\alpha$-logarithmic negativity can be experimentally measured via witnessing. This is true for the logarithmic negativity \cite{FB05} (see also \cite{SHR17}) and the $\kappa$-entanglement \cite{WW20}, and so it remains open to determine the answer for $\alpha \in (1,\infty)$.

We also think that it is interesting to explore generalizations of quantum
information measures that accept Hermitian operators as input, rather than
just positive semi-definite ones, the latter being the traditional approach in
quantum information theory. The former approach has been employed fruitfully
not only here, but also in recent work that gave exact characterizations of
the one-shot distillable entanglement of a bipartite quantum state
\cite{FWTD17} and the one-shot distillable coherence of a quantum state
\cite{BFWA18}.

\textit{Note}: After a preprint of our paper appeared online as \cite{WW19},
we noticed another work \cite{CMT19} that proposed different generalizations
of the logarithmic negativity.

\begin{acknowledgments}
We are grateful to Nilanjana Datta and Gilad Gour for feedback during the Workshop \textit{Mathematical Aspects in Current Quantum Information Theory 2019}
(MAQIT~2019)
held at Seoul National University, Korea.
Part of this work was done when
XW was at the University of Maryland.
MMW acknowledges
support from the National Science Foundation under Award nos.~1350397 and 1907615.
\end{acknowledgments}

\bibliographystyle{unsrt}
\bibliography{Ref}

\end{document}